\documentclass[
BCOR2pt, 
captions=nooneline, 
bibliography=totoc, 
numbers=noenddot, 
parskip=half, 
headings=normal, 
abstracton 
]{scrartcl} 

\usepackage[USenglish]{babel} 
\usepackage{graphicx} 
\usepackage{tikz}
\usepackage{framed}
\usepackage{floatrow}
\usepackage{remreset} 
\usepackage[nouppercase]{scrpage2} 
\usepackage{faktor}
\usepackage{amsmath} 
\usepackage{amssymb} 
\usepackage[thmmarks, 
amsmath, 
hyperref 
]{ntheorem} 
\usepackage{bm} 
\usepackage{bbm} 
\usepackage{enumitem} 
\usepackage[hang]{subfigure} 
\usepackage{wrapfig} 
\usepackage{tabularx} 
\usepackage{color}
\usepackage{dcolumn} 
\usepackage{booktabs} 
\usepackage{listings} 
\usepackage{psfrag} 
\usepackage[pdftex, 
setpagesize=false, 
pdfborder={0 0 0}, 
pdfpagemode=UseOutlines, 
]{hyperref} 

\makeatletter
\newcommand\mytoday{\number\year-\ifcase\month\or 01\or 02\or 03\or 04\or 05\or 06\or 07\or 08\or 09\or 10\or 11\or 12\fi-\ifcase\day\or 01\or 02\or 03\or 04\or 05\or 06\or 07\or 08\or 09\or 10\or 11\or 12\or 13\or 14\or 15\or 16\or 17\or 18\or 19\or 20\or 21\or 22\or 23\or 24\or 25\or 26\or 27\or 28\or 29\or 30\or 31\fi} 
\makeatother
\setkomafont{sectioning}{\normalcolor\bfseries} 
\clearscrheadfoot 
\automark[section]{subsection} 
\setkomafont{captionlabel}{\bfseries} 
\newcolumntype{d}[2]{D{.}{.}{#1.#2}} 
\newcommand*{\abstractnoindent}{} 
\let\abstractnoindent\abstract
\renewcommand*{\abstract}{\let\quotation\quote\let\endquotation\endquote
\abstractnoindent}
\makeatletter
\renewcommand{\p@enumii}[1]{\theenumi(#1)}
\makeatother


\theoremstyle{break} 
\theoremheaderfont{\bfseries}
\theorembodyfont{}
\theoremseparator{}
\newtheorem{definition}{Definition}[section] 
\newtheorem{lemma}[definition]{Lemma}
\newtheorem{theorem}[definition]{Theorem}

\newtheorem{corollary}[definition]{Corollary}
\newtheorem{remark}[definition]{Remark}
\newtheorem{example}[definition]{Example}

\theoremstyle{nonumberbreak} 
\theoremsymbol{$\Box$}
\newtheorem{proof}{Proof}

\renewcommand*{\L}{\mathcal{L}}

\newcommand*{\IP}{\mathbb{P}}
\newcommand*{\IE}{\mathbb{E}}

\newcommand*{\IN}{\mathbb{N}}

\newcommand*{\IQ}{\mathbb{Q}}

\newcommand*{\F}{\mathfrak{F}}
\newcommand*{\G}{\mathcal{G}}

\usetikzlibrary{snakes}
\usetikzlibrary{arrows,shapes,positioning}
\usetikzlibrary{decorations.markings}
\tikzstyle arrowstyle=[scale=2]
\tikzstyle directed=[postaction={decorate,decoration={markings,
    mark=at position 1 with {\arrow[arrowstyle]{stealth}}}}]

\begin{document}

\renewcommand{\figurename}{Fig.}

\thispagestyle{plain}
	\begin{center}
		{\bfseries\Large Canonical spectral representation for exchangeable max-stable sequences}
		\par\bigskip
		\vspace{1cm}
		
		{\Large Jan-Frederik Mai}\\
		\vspace{0.2cm}
		{Technische Universit\"at M\"unchen}\\
		{Parkring 11, 85478 Garching-Hochbr\"uck, Germany}\\
		{email: mai@tum.de}\\

	\end{center}

\textit{This is a pre-print of an article published in Extremes. The final authenticated version is available online at: https://doi.org/10.1007/s10687-019-00361-3.}

\section*{Abstract}
The set $\mathfrak{L}$ of infinite-dimensional, symmetric stable tail dependence functions associated with exchangeable max-stable sequences of random variables with unit Fr\'echet margins is shown to be a simplex. Except for a single element, the extremal boundary of $\mathfrak{L}$ is in one-to-one correspondence with the set $\F_1$ of distribution functions of non-negative random variables with unit mean. Consequently, each $\ell \in \mathfrak{L}$ is uniquely represented by a pair $(b,\mu)$ of a constant $b$ and a probability measure $\mu$ on $\F_1$. A canonical stochastic construction for arbitrary exchangeable max-stable sequences and a stochastic representation for the Pickands dependence measure of finite-dimensional margins of $\ell$ are immediate corollaries. As by-products, a canonical analytical description and an associated canonical Le Page series representation for non-decreasing stochastic processes that are strongly infinitely divisible with respect to time are obtained. 

\textbf{Keywords:} exchangeable sequence, max-stable sequence, stable tail dependence function, extreme-value copula, strong IDT process, Pickands representation

\section{Motivation and mathematical background}
Before we start, we clarify some notation. We will denote the indexing argument of a stochastic process $f$ as a subindex by writing $f_t$, instead of $f(t)$. This is in order to emphasize that $f_t$ is a random variable, whereas when writing $f(t)$ we mean the (non-random) value of a deterministic function $f$ in the variable $t$. Further, for $x>0$ we define $x/0:=\infty$ in order to simplify notation. We denote by $[0,\infty)_{00}^{\IN}$ the set of sequences $\vec{t}=(t_1,t_2,\ldots)$ with non-negative members that are eventually zero, i.e.\ $t_k=0$ for almost all $k \in \IN$. 
\par
An (infinite) sequence of random variables is called \emph{exchangeable} if its probability distribution is invariant with respect to permutations of finitely many, but arbitrarily many, constituents, see \cite{aldous85} for a textbook treatment. Let $\vec{Y}=(Y_1,Y_2,\ldots)$ be an exchangeable sequence of random variables with $\IE[Y_1]=1$ and such that $\min_{k \geq 1}\{Y_k/t_k\}$ has an exponential distribution with rate $\ell(\vec{t})$ for arbitrary $\vec{t} \in [0,\infty)_{00}^{\IN}$ except the zero sequence (or this one included when interpreting the exponential distribution with rate zero as an atom at $\infty$). The sequence $\vec{Y}$ is said to be \emph{min-stable multivariate exponential} and its law is uniquely described by the function $\ell$, which is called the \emph{stable tail dependence function} of the sequence $\vec{Y}$. Indeed,
\begin{gather*}
\IP(\vec{Y}>\vec{t}) = \IP\big(\min_{k \geq 1}\{Y_k/t_k\}>1\big)= \exp\big\{ -\ell(\vec{t})\big\},\quad \vec{t} \in[0,\infty)_{00}^{\IN}, 
\end{gather*}
with the first ``$>$''-sign understood componentwise. The function $\ell$ is symmetric in its arguments by exchangeability of $\vec{Y}$. Our goal is to determine the shape of the set of all symmetric stable tail dependence functions $\ell:[0,\infty)_{00}^{\IN} \rightarrow [0,\infty)$, which we denote by $\mathfrak{L}$ in the sequel. The sequence $1/\vec{Y}=(1/Y_1,1/Y_2,\ldots)$ is called \emph{max-stable with unit Fr\'echet margins}, which makes the probability law of $\vec{Y}$ interesting for the field of multivariate extreme-value theory, for background the interested reader is referred to \cite{segers12} and \cite[Chapter 6]{joe97}. Concretely, the set of all $d$-variate functions $C_{\ell}:[0,1]^d\rightarrow [0,1]$, defined by
\begin{gather*}
C_{\ell}(u_1,\ldots,u_d):=\exp\Big\{ - \ell\big( -\log(u_1),\ldots,-\log(u_d),0,0,\ldots\big)\Big\},\quad \ell \in \mathfrak{L},
\end{gather*}
constitutes a (proper) subfamily of exchangeable, $d$-dimensional \emph{extreme-value copulas}\footnote{For background, the interested reader is referred to \cite{gudendorf09}.}. More precisely, an exchangeable $d$-dimensional extreme-value copula $C$ is of the form $C_{\ell}$ if and only if there is an infinite exchangeable sequence $\vec{U}=(U_1,U_2,\ldots)$ on some probability space $(\Omega,\G,\IP)$ such that $C$ equals the distribution function of $(U_1,\ldots,U_d)$. By virtue of De Finetti's Theorem this is the case if and only if there exists a sub-$\sigma$-algebra $\mathcal{T} \subset \G$ such that $U_1,\ldots,U_d$ (and $U_{d+1},\,U_{d+2},\ldots$ as well) are independent and identically distributed (iid) conditioned on $\mathcal{T}$. We recall from \cite[p.\ 26, Corollary 3.12]{aldous85} that $\mathcal{T}$ almost surely coincides with the tail-$\sigma$-field $\cap_{k \geq 1}\sigma(U_k,\,U_{k+1},\ldots)$ of $\vec{U}$. It is educational to remark, however, that there are exchangeable $d$-dimensional extreme-value copulas that are not of the form $C_{\ell}$, because in general the notion of ``infinite exchangeability'' (or, more loosely, ``conditionally iid'') is stronger than that of finite ($d$-dimensional) exchangeability. In analytical terms, a $d$-margin of some $\ell \in \mathfrak{L}$ is always a symmetric stable tail dependence function, but not every symmetric, $d$-variate stable tail dependence function is a $d$-margin of some $\ell \in \mathfrak{L}$. 
\par
It is well known at least since \cite{dehaan84} that $\ell$ can be represented as
\begin{gather}
\ell(\vec{t}) = -\log\big\{ \IP(\vec{Y}>\vec{t}) \big\}=\IE\Big[ \max_{k \geq 1}\{t_k\,X_k\}\Big],
\label{spectralrepr}
\end{gather}
for some sequence $\vec{X}=(X_1,X_2,\ldots)$ of random variables with finite means\footnote{Even though we are only interested in distributional statements throughout, for the sake of a more intuitive exposition we find it sometimes convenient to express formulas like (\ref{spectralrepr}) in probabilistic notation with probability measure $\IP$ and expectation $\IE$ (as compared to writing integrals), with the generic random objects $\vec{X},\vec{Y}$ being viewed as defined on some generic probability space $(\Omega,\G,\IP)$, on which we do not necessarily work.}. As an example\footnote{This is the example on page 1198 in \cite{dehaan84}.} for the representation (\ref{spectralrepr}), if $\vec{Y}$ has independent components which all have the unit exponential distribution, that is $\ell(\vec{t})=\sum_{k \geq 1}t_k$, the probability law of $\vec{X}$ can be defined via a vector of discrete probabilities $\vec{p}=(p_1,p_2,\ldots)$ as
\begin{gather}
\IP\Big(\vec{X} = \frac{\vec{e}_k}{p_k} \Big) = p_k > 0,\quad k \geq 1,\quad \sum_{k \geq 1}p_k=1,
\label{repr_nonunique}
\end{gather}
where $\vec{e}_k$ denotes the sequence with all members equal to zero except for the $k$-th. The representation (\ref{spectralrepr}) of a stable tail dependence function is called a \emph{spectral representation}. As (\ref{repr_nonunique}) shows, it is not unique in general (i.e.\ different $\vec{X}$ can imply the same $\ell$, hence $\vec{Y}$). Furthermore, even though $\vec{Y}$ is assumed to be exchangeable in the present work, $\vec{X}$ needs not be exchangeable and, in fact, the proof in \cite{dehaan84} constructs $\vec{X}$ from $\vec{Y}$ in such a way that $\vec{X}$ is not exchangeable (the particular example (\ref{repr_nonunique}) demonstrates this). Conversely, however, the spectral representation (\ref{spectralrepr}) can be used to construct models for $\vec{Y}$ by choosing convenient models for $\vec{X}$ that allow the expected value in (\ref{spectralrepr}) to be computed in closed form, as highlighted in \cite{segers12}. If one pursues this strategy and starts with an exchangeable $\vec{X}$, one obtains an exchangeable sequence $\vec{Y}$, but to the best of our knowledge it is an open question (solved by the present article) whether all exchangeable min-stable multivariate exponential $\vec{Y}$ can be obtained in such a way.  
\par
We denote by $\ell^{(d)}$ the restriction of $\ell$ to the first $d \in \IN$ components, i.e.\ $\ell^{(d)}$ determines the law of $(Y_1,\ldots,Y_d)$. For stable tail dependence functions in finite dimensions, such as $\ell^{(d)}$, there exist different methods to obtain uniqueness of the spectral representation (\ref{spectralrepr}) by imposing certain restrictions on the law of $\vec{X}$. The most prominent one is the \emph{Pickands representation}, named after \cite{pickands81}, see also \cite{dehaan77,ressel13}, which states that if $\ell^{(d)}$ is the stable tail dependence function associated with some min-stable multivariate exponential random vector $\vec{Y}^{(d)}=(Y_1,\ldots,Y_d)$, then there is a random vector $\vec{X}^{(d)}=(X^{(d)}_1,\ldots,X^{(d)}_d)$, uniquely determined in law, which takes values on the $d$-dimensional unit simplex $S_d:=\{\vec{q} \in [0,1]^d\,:\,q_1+\ldots+q_d=1\}$ and satisfying $\IE[X^{(d)}_k]=1/d$ for all $k=1,\ldots,d$, such that
\begin{gather}
 \ell^{(d)}(\vec{t}) = d\,\IE\big[ \max\{t_1\,X^{(d)}_1,\ldots,t_d\,X^{(d)}_d\}\big].
\label{Pickandsrepr}
\end{gather}
In our infinite-dimensional setting, even though we assume exchangeability of $\vec{Y}=(Y_1,Y_2,\ldots)$, an unfortunate aspect of the Pickands representation is that the relation between the laws of $\vec{X}^{(d)}$ and $\vec{X}^{(d+1)}$ is not easy to understand, in particular $\vec{X}^{(d)}$ is not a re-scaled $d$-margin of $\vec{X}^{(d+1)}$, like one might naively hope on first glimpse. Consequently, describing the infinite-dimensional, symmetric stable tail dependence function $\ell$ in terms of the collection of its finite-dimensional Pickands measures is neither easily accomplished nor convenient or algebraically natural.
\par
In the main body of this article, we derive a natural and convenient spectral representation for symmetric stable tail dependence functions. To wit, each $\ell \in \mathfrak{L}$ can be represented as
\begin{gather}
\ell(\vec{t}) = b\,\sum_{k \geq 1}t_k + (1-b)\,\IE\big[ \max_{k \geq 1}\{t_k\,X_k\}\big],\quad \vec{t} \in [0,\infty)^{\IN}_{00},
\label{STDFrepr_prob}
\end{gather}
with a constant $b \in [0,1]$ and an exchangeable sequence $\vec{X}$ satisfying $\IE[X_1]=1$ (hence $\IE[X_k]=1$ for each $k \geq 1$). Whereas the constant $b$ is unique, the law of $\vec{X}$ is still not unique in general, but it becomes unique if we postulate in addition that the conditional mean of $\vec{X}$, that is $\lim_{n \rightarrow \infty}\frac{1}{n}\sum_{k=1}^{n}X_k$, is identically equal to one. In particular, the stable tail dependence function $\ell(\vec{t})=\sum_{k \geq 1}t_k$, corresponding to independent members in $\vec{Y}$, cannot be represented via an exchangeable $\vec{X}$.  However, the canonical representation (\ref{STDFrepr_prob}) shows that the independence case occupies an isolated role in this regard.
\par
By virtue of De Finetti's Theorem, see \cite{definetti31,definetti37} and \cite[p.\ 19 ff]{aldous85}, studying the law of the exchangeable sequence $\vec{Y}$ is tantamount to a study of the law of a random distribution function $F=\{F_t\}_{t \geq 0}$ that is defined by $F_t:=\IP(Y_1 \leq t\,|\,\mathcal{T})$, with $\mathcal{T}=\cap_{k \geq 1}\sigma(Y_k,\,Y_{k+1},\ldots)$ denoting the tail-$\sigma$-field of $\vec{Y}$. A result of \cite{maischerer13} shows that the stochastic process $H:=-\log(1-F)$ is\footnote{\cite{molchanov18} call the ``strong IDT'' processes ``time-stable'' processes, but we prefer to stick with the original nomenclature.} \emph{strongly infinitely divisible with respect to time (strong IDT)}, meaning that
\begin{gather*}
\{H_t\}_{t \geq 0} \stackrel{d}{=} \big\{ H^{(1)}_{\frac{t}{n}}+\ldots+ H^{(n)}_{\frac{t}{n}}\big\}_{t \geq 0},\quad \forall n \in \IN,
\end{gather*}
where $\stackrel{d}{=}$ denotes equality in law and $H^{(i)}$ are independent copies of $H$. Conversely, given a non-decreasing, right-continuous strong IDT process $H$ and an independent sequence $\{\eta_k\}_{k \geq 1}$ of iid unit exponential variables, the exchangeable, min-stable multivariate exponential sequence $\vec{Y}$ can be represented as
\begin{gather}
Y_k:=\inf\{t>0\,:\,H_t>\eta_k\},\quad k \in \IN,
\label{ciiddef}
\end{gather}
establishing a canonical stochastic representation, which is conditionally iid in the sense of De Finetti's Theorem. If $H$ is normalized to satisfy $\IE[\exp(-H_1)]=\exp(-1)$, it follows that $\IE[Y_1]=1$, so that the function 
\begin{gather*}
\ell(\vec{t}):=-\log\{\IP(\vec{Y}>\vec{t})\}=-\log\Big\{\IE\Big[e^{-\sum_{k \geq 1}H_{t_k}}\Big]\Big\},\quad \vec{t} \in [0,\infty)^{\IN}_{00},
\end{gather*}
lies in $\mathfrak{L}$.
\par 
In fact, in our proof of (\ref{STDFrepr_prob}) we rely heavily on the concept of strong IDT processes, for which \cite{molchanov18} recently have derived a convenient series representation, which we make use of. Translating the analytical result (\ref{STDFrepr_prob}) on symmetric stable tail dependence functions into the language of these processes then implies that each non-decreasing strong IDT process is uniquely determined by a triplet $(b,c,\mu)$ of constants $b \geq 0$, $c>0$ and a probability measure $\mu$ on the set of distribution functions of non-negative random variables with unit mean, as we will see. 
\par
Regarding the organization of the remaining article, we prove and discuss the main result (\ref{STDFrepr_prob}) in Section \ref{sec_proof}, and we conclude in Section \ref{sec_concl}.
\section{The structure of $\mathfrak{L}$}\label{sec_proof}
We denote by $\ell_{\Pi}(\vec{t}):=\sum_{k \geq 1}t_k$ the stable tail dependence function associated with an iid sequence of unit exponentials. We denote by $\F$ (resp.\ $\F_1$) the set of all distribution functions of non-negative random variables with finite (resp.\ unit) mean. Then with $F \in \F_1$ the function
\begin{gather*}
\ell_F(\vec{t}):=\int_0^{\infty}1-\prod_{k \geq 1}F\Big( \frac{s}{t_k}\Big)\,\mathrm{d}s,\quad \vec{t} \in [0,\infty)^{\IN}_{00},
\end{gather*}
defines a symmetric stable tail dependence function, which is investigated thoroughly in \cite{mai17}. In this definition, implicitly we mean $F(\infty)=1$ for those $t_k$ that are zero. We remark that $\ell_F$ has the stochastic representation $\ell_F(\vec{t})=\IE[\max_{k \geq 1}\{t_k\,X_k\}]$, where $\vec{X}=(X_1,X_2,\ldots)$ is an iid sequence drawn from $F$. We seek to show that $\mathfrak{L}$ (equipped with the topology of pointwise convergence) is a simplex with extremal boundary $\partial_e \mathfrak{L}=\{\ell_{\Pi}\}\cup \{\ell_F\,:\,F \in \F_1\}$, which is the main contribution of the present work, see Theorem \ref{thm_main} and Corollary \ref{cor_uni} below.
\par
The key idea of the presented proof relies on the aforementioned link to strong IDT processes, found in \cite[Theorem 5.3]{maischerer13}. In a recent article, \cite[Theorem 4.2]{molchanov18} show that a non-negative\footnote{More generally, \cite{molchanov18} consider strong IDT processes without Gaussian component, but only the subclass of non-negative ones is of interest in the present article.}, c\`adl\`ag, stochastically continuous, strong IDT process without Gaussian component admits a LePage series representation
\begin{gather}
\{H_t\}_{t \geq 0} \stackrel{d}{=} \Big\{b\,t+\sum_{k \geq 1}f^{(k)}_{\frac{t}{\epsilon_1+\ldots+\epsilon_k}}\Big\}_{t \geq 0},\quad t \geq 0,
\label{molchanov_repr}
\end{gather}
where $\{f^{(k)}\}_{k \geq 1}$ is a sequence of independent copies of a non-vanishing c\`adl\`ag stochastic process $f=\{f_t\}_{t \geq 0}$ with $f_0=0$ (denote the space of all such functions by $\mathfrak{D}$ in the sequel) and, independently, $\{\epsilon_k\}_{k \geq 1}$ is a list of iid unit exponentials. Furthermore, the process $f$ satisfies 
\begin{gather}
\int_{\mathfrak{D}} \int_0^{\infty}\min\{1,|f(u)|\}\,\frac{\mathrm{d}u}{u^2}\,\gamma(\mathrm{d}f)<\infty,
\label{intcondition}
\end{gather}
with $\gamma$ denoting the probability law of $f$ on $\mathfrak{D}$. In general, the law of $f$ in this representation of a non-negative strong IDT process is non-unique. However, the following auxiliary lemma shows that if $H$ is non-decreasing, we can at least learn that $f$ is non-decreasing as well. This proof is the most technical step towards Theorem \ref{thm_main} below, and it is a result of independent interest as well.
\begin{lemma}[Non-decreasing strong IDT processes] \label{lemma_molchanov}
If $H$ is a non-decreasing, right-continuous strong IDT process, the stochastic process $f$ of any LePage series representation (\ref{molchanov_repr}) is necessarily non-decreasing and $b \geq 0$.
\end{lemma}
\begin{proof}
First notice that non-decreasingness, right-continuity, and the strong IDT property imply stochastic continuity of $H$. This follows from the fact that for each fixed $t>0$ the random variable $H_{t-}:=\lim_{x \uparrow t}H_x$ exists in $[0,\infty]$ by non-decreasingness and has the same infinitely divisible law as the random variable $H_t$ (since $\IE[\exp(-x\,H_u)]=\exp\{-t\,\Psi_H(u)\}$ for some Bernstein function\footnote{See \cite{schilling10} for a textbook treatment on Bernstein functions.} $\Psi_H$, see \cite[Lemma 3.7]{maischerer13}). Thus, the non-negative random variable $H_t-H_{t-}$ has zero expectation and is thus identically equal to zero. On the other hand, right-continuity implies that $H_{t+}:=\lim_{x \downarrow 0}H_x=H_t$ almost surely. Consequently, the stochastic continuity assumption of \cite[Theorem 4.2]{molchanov18} is satisfied, hence there is a LePage series representation.
\par
Consider a probability space $(\Omega,\G,\IP)$, on which $H$ is defined by the right-hand side of (\ref{molchanov_repr}). We first prove that $f$ is non-decreasing. The heuristic idea is to consider
\begin{gather*}
H_{\epsilon_1\,t} = f^{(1)}_t+ b\,\epsilon_1\,t+ \sum_{k \geq 2}f^{(k)}_{\frac{\epsilon_1\,t}{\epsilon_1+\ldots+\epsilon_k}},
\end{gather*}
and argue that there is a positive probability that $\epsilon_1$ is so small that $f^{(1)}_t$ is the dominating part on the right-hand side, from which one concludes that a violation of the non-decreasingness of $f^{(1)}$ on some interval $[x_1,x_2]$ would imply a violation of the non-decreasingness of  $H$ on $[x_1/\epsilon_1,x_2/\epsilon_1]$. 
\par
To fill this intuitive idea with some mathematical rigor is what's done in the sequel. We assume a violation of non-decreasingness of $f^{(1)}$, which means that there exists $\epsilon>0$ and $0 \leq x_1<x_2<\infty$ such that $\IP(A_{f})>0$ for the event
\begin{gather*}
A_{f}:=\Big\{ f^{(1)}\mbox{ not non-decreasing on }[x_1,x_2] \mbox{ and }\inf_{x \in [x_1,x_2]}\{f^{(1)}_x-f^{(1)}_{x_1}\} \leq -\epsilon\Big\}.
\end{gather*}
Our goal is to show that this implies a violation of the non-decreasingness of $H$. For later use we define the $\sigma$-algebra $\mathcal{H}:=\sigma(\epsilon_k,f^{(k)}\,:\,k \geq 2)$ generated by all involved stochastic objects except for $\epsilon_1,f^{(1)}$. 
\par
For a moment consider independent copies $\{g^{(k)}\}_{k \geq 1}$ of $g:=|f|$ and for each $x \in (0,1],\,t \geq 0$, let
\begin{align*}
\tilde{H}_t&:=\sum_{k \geq 2}g^{(k)}_{\frac{t}{\epsilon_2+\ldots+\epsilon_k}},\quad \tilde{H}^{(x)}_t:=\sum_{k \geq 2}g^{(k)}_{\frac{x\,t}{x+\epsilon_2+\ldots+\epsilon_k}},\\
 \hat{H}^{(x)}_t&:=\sum_{k \geq 2}g^{(k)}_{\frac{x\,t}{\epsilon_2+\ldots+\epsilon_k}}\,1_{\{\epsilon_2+\ldots+\epsilon_k > x\}}.
\end{align*}
By \cite[Theorem 4.2]{molchanov18}, the process $\tilde{H}$ is non-negative, strong IDT, c\`adl\`ag, and satisfies $\tilde{H}_0=0$. Right-continuity in zero implies for each $x \in (0,1]$ that the first exit time of $\{\tilde{H}_{x\,t}\}_{t \geq 0}$ from the interval $[0,\epsilon/4]$ is almost surely positive, i.e.\
\begin{gather*}
\tilde{T}_x:=\inf\{t>0\,:\,\tilde{H}_{x\,t}>\epsilon/4\}>0.
\end{gather*}
Furthermore, it is obvious from the definition that $\tilde{T}_x=\tilde{T}_1/x$, which implies that $\lim_{x \searrow 0}\tilde{T}_x = \infty$ almost surely. We use the Laplace functional formula \cite[Proposition 3.6]{resnick87} for Poisson random measure to observe for $d \in \IN$ and $t_1,y_1,\ldots,t_d,y_d \geq 0$ arbitrary that
\begin{align*}
\IE\Big[e^{-\sum_{i=1}^{d}y_i\,\tilde{H}^{(x)}_{t_i}}\Big] &= \exp\Big( -\int_{\mathfrak{D}}\int_0^{\infty}1-e^{\sum_{i=1}^{d}y_i\,|f|(x\,t_i/(x+s))}\,\mathrm{d}s\,\gamma(\mathrm{d}f)\Big)\\
& = \exp\Big( -x\,\int_{\mathfrak{D}}\int_0^{1}1-e^{\sum_{i=1}^{d}y_i\,|f|(u\,t_i)}\,\frac{\mathrm{d}u}{u^2}\,\gamma(\mathrm{d}f)\Big),
\end{align*}
having applied the substitution $u=x/(x+s)$. An analogous computation with the substitution $u=x/s$ shows that
\begin{align*}
\IE\Big[e^{-\sum_{i=1}^{d}y_i\,\hat{H}^{(x)}_{t_i}}\Big] &= \exp\Big( -\int_{\mathfrak{D}}\int_0^{\infty}1-e^{\sum_{i=1}^{d}y_i\,|f|(x\,t_i/s)\,1_{\{s>x\}}}\,\mathrm{d}s\,\gamma(\mathrm{d}f)\Big)\\
& = \exp\Big( -x\,\int_{\mathfrak{D}}\int_0^{1}1-e^{\sum_{i=1}^{d}y_i\,|f|(u\,t_i)}\,\frac{\mathrm{d}u}{u^2}\,\gamma(\mathrm{d}f)\Big),
\end{align*}
which implies that $\{\tilde{H}^{(x)}_t\}_{t \geq 0}$ has the same law as $\{\hat{H}_t^{(x)}\}_{t \geq 0}$. This implies that the first exit time of $\tilde{H}^{(x)}$ from $[0,\epsilon/4]$ is equal in law to that of $\hat{H}^{(x)}$. The process $\hat{H}^{(x)}$ evidently satisfies
\begin{gather*}
\hat{H}^{(x)}_t = \sum_{k \geq 2}g^{(k)}_{\frac{x\,t}{\epsilon_2+\ldots+\epsilon_k}}\,1_{\{\epsilon_1+\ldots+\epsilon_k > x\}} \leq \sum_{k \geq 2}g^{(k)}_{\frac{x\,t}{\epsilon_2+\ldots+\epsilon_k}}=\tilde{H}_{t\,x}.
\end{gather*}
For its first exit time from $[0,\epsilon/4]$ this gives the lower bound
\begin{gather*}
\inf\{t>0\,:\,\hat{H}^{(x)}_t>\epsilon/4\} \geq \inf\{t>0\,:\,\tilde{H}_{x\,t}>\epsilon/4\} = \tilde{T}_x,
\end{gather*}
which was shown to converge to infinity almost surely as $x \searrow 0$. Since $\tilde{H}^{(x)}\stackrel{d}{=}\hat{H}^{(x)}$, the first exit time of $\tilde{H}^{(x)}$ from $[0,\epsilon/4]$ is thus also shown to converge to infinity as $x \searrow 0$.
\par
Now the process of interest for us is 
\begin{gather*}
{H}^{(x)}_t:=\sum_{k \geq 2}f^{(k)}_{\frac{x\,t}{x+\epsilon_2+\ldots+\epsilon_k}},\quad t \geq 0,
\end{gather*}
which satisfies $|H^{(x)}_t| \leq \tilde{H}^{(x)}_t$ for all $t$. Consequently, the first exit time $T_x$ of $H^{(x)}$ from the interval $[-\epsilon/4,\epsilon/4]$ is almost surely larger than that of $\tilde{H}^{(x)}$, which was shown above to converge to infinity almost surely as $x \searrow 0$. Consequently, we find an $\mathcal{H}$-measurable (notice that $H^{(x)}$ is $\mathcal{H}$-measurable) random variable $Z>0$ such that $T_x \geq x_2$ for all $x \leq Z$. In particular, on the event $\{\epsilon_1 \leq Z\}$ we have that $T_{\epsilon_1} \geq x_2$ and hence $\sup_{t \in [0,x_2]}|H^{(\epsilon_1)}_t| \leq \epsilon/4$. Finally, on the event $A_{\epsilon}:=\{\epsilon_1 < \epsilon/(4\,x_2\,|b|)\}$ we have $|b\,\epsilon_1\,t| \leq \epsilon/4$ for all $t \leq x_2$. Notice that $A_{\epsilon}$ has positive probability (possibly equal to one if $b=0$). Summing up all terms, we show that the event
\begin{align*}
A_H&:= \Big\{ \{H_{\epsilon_1\,t}\}_{t \geq 0} \mbox{ not non-decreasing on }[x_1,x_2]\\
& \qquad \mbox{ and }\inf_{x \in [x_1,x_2]}\{H_{\epsilon_1\,x}-H_{\epsilon_1\,x_1}\} \leq -\frac{\epsilon}{4}\Big\} 
\end{align*}
has positive probability. To this end, by construction $(A_{f} \cap A_{\epsilon} \cap \{\epsilon_1 \leq Z\}) \subset A_H$, since on this set we have for $x \in [x_1,x_2]$ that
\begin{align*}
\inf_{x \in [x_1,x_2]}\{H_{\epsilon_1\,x}-H_{\epsilon_1\,x_1}\} &=  \hspace{-.4cm}\inf_{x \in [x_1,x_2]}\{f^{(1)}_x-f^{(1)}_{x_1}+\underbrace{b\,\epsilon_1\,(x-x_1)}_{\leq \epsilon/4} + \hspace{-.4cm}\underbrace{H^{(\epsilon_1)}_x-H^{(\epsilon_1)}_{x_1}}_{\leq |H_x^{(\epsilon_1)}|+|H_{x_1}^{(\epsilon_1)}| \leq \epsilon/2}\hspace{-.4cm}\} \\
& \leq \underbrace{\inf_{x \in [x_1,x_2]}\{f^{(1)}_x-f^{(1)}_{x_1}\}}_{\leq -\epsilon}+\frac{3\,\epsilon}{4} \leq -\frac{\epsilon}{4}.
\end{align*}
Hence,
\begin{align*}
\IP(A_H) & \geq \IE[1_{A_{f}}\,1_{A_{\epsilon}}\,1_{\{\epsilon_1 \leq Z\}}]\\
& = \IP(A_f)\,\IE[\IE[1_{A_{\epsilon}}\,1_{\{\epsilon_1 \leq Z\}}\,|\,\mathcal{H}]] = \IP(A_f)\,\IE\Big[ 1-e^{-\min\{Z,\epsilon/(4\,|b|\,x_2)\}}\Big]>0.
\end{align*}
That the last expression is strictly positive follows from the fact that the random variable $\min\{Z,\epsilon/(4\,|b|\,x_2)\}$ is not almost surely zero (it is even almost surely positive). Since $A_H$ has positive probability, $H$ cannot be non-decreasing almost surely, hence the assumption was wrong and $f$ needs to be non-decreasing.
\par
Next, we prove that $b \geq 0$. To this end, we know that $H_1$ is non-negative and infinitely divisible, consequently it has a non-negative drift $b_H \geq 0$ in its L\'evy-Khinchin representation, see \cite{bertoin99,schilling10} for background. Also, the stochastic process 
\begin{gather*}
\tilde{H}_t:=\sum_{k \geq 1}f^{(k)}_{\frac{t}{\epsilon_1+\ldots+\epsilon_k}}, \quad t \geq 0, 
\end{gather*}
is strong IDT by \cite[Theorem 4.2]{molchanov18}, hence $\tilde{H}_1$ is infinitely divisible. But by what we have shown, each $f^{(k)}$ is non-negative almost surely, so $\tilde{H}_1 \geq 0$. Since $f \geq 0$ by what we have just shown, the Bernstein function $\Psi_{\tilde{H}}$ associated with $\tilde{H}_1$ via $\Psi_{\tilde{H}}(x)=-\log(\IE[\exp(-x\,\tilde{H}_1)])$ can be computed using the Laplace functional formula for Poisson random measure, cf.\ \cite[Proposition 3.6]{resnick87}, which yields
\begin{gather*}
\Psi_{\tilde{H}}(x) = \int_{\mathfrak{D}} \int_0^{\infty}1-e^{-x\,f(u)}\,\frac{\mathrm{d}u}{u^2}\,\gamma(\mathrm{d}f).
\end{gather*}
By non-negativity of $f$, we get for all $x \geq 1$ and all $u>0$ the estimate $(1-\exp(-x\,f(u)))/x \leq \min\{1/x,f(u)\} \leq \min\{1,f(u)\}$. By the dominated convergence theorem, using (\ref{intcondition}), we may thus conclude that $\Psi_{\tilde{H}}(x)/x$ converges to zero as $x \rightarrow \infty$. Consequently, the random variable $\tilde{H}_1$ has no drift in its L\'evy-Khinchin representation. This implies $b=b_H \geq 0$.
\end{proof}

We are now in the position to derive the main contribution of the present article. We denote by $M_+^{1}( \mathfrak{E})$ the set of Radon probability measures on some Hausdorff space $\mathfrak{E}$. In particular, we consider the sets $\F$ and $\F_1$ equipped with the topology induced by convergence in distribution of the associated non-negative random variables (aka weak convergence of distribution functions). This topology is well known to be metrizable, hence is Hausdorff, see \cite{sibley71}.
\begin{theorem}[The structure of $\mathfrak{L}$] \label{thm_main}
Let $\ell \in \mathfrak{L}$, not equal to $\ell_{\Pi}$. There exists a pair $(b,\mu) \in [0,1) \times M_+^{1}(\F_1)$ such that $\ell = b\,\ell_{\Pi}+ (1-b)\,\int_{\F_1}\ell_F\,\mu(\mathrm{d}F)$.
\end{theorem}
\begin{proof}
Given $\ell \in \L$, there exists an exchangeable sequence $\vec{Y}=(Y_1,Y_2,\ldots)$ of random variables on a probability space $(\Omega,\G,\IP)$ such that $\IP(\vec{Y}>\vec{t})=\exp(-\ell(\vec{t}))$ for $\vec{t} \in [0,\infty)_{00}^{\IN}$. By \cite[Theorem 5.3]{maischerer13}, the stochastic process $H_t:=-\log\big( \IP(Y_1>t\,|\,\mathcal{T})\big)$, $t \geq 0$, with $\mathcal{T}$ the tail-$\sigma$-field of $\vec{Y}$, is strong IDT, non-decreasing and not identically equal to $H_t=t$ (since $\ell \neq \ell_{\Pi}$). Lemma \ref{lemma_molchanov} proves existence of $b \geq 0$ and a non-vanishing, right-continuous, non-decreasing stochastic process $f=\{f_t\}_{t \geq 0}$ with $f_0=0$ (denote the space of all such functions by $\mathfrak{D}_+$ in the sequel) such that (\ref{molchanov_repr}) holds. Denoting the probability law of $f$ by $\gamma$, (\ref{intcondition}) now reads
\begin{gather}
\int_{\mathfrak{D}_+} \int_0^{\infty}\min\{1,f(u)\}\,\frac{\mathrm{d}u}{u^2}\,\gamma(\mathrm{d}f)<\infty.
\label{DCTcond}
\end{gather}
From the properties of $f$ (non-decreasingness, right-continuity and $f(0)=0$) we conclude that the function
\begin{gather*}
\tilde{F}_t:= e^{-\lim_{x \downarrow t} f_{1/x}} ,\quad t \geq 0,
\end{gather*}
is almost surely the distribution function of some non-negative random variable, which is not identically zero (since $f$ is non-vanishing). In the sequel, we denote the probability measure of $\tilde{F}$ by $\tilde{\mu}$. We get from (\ref{DCTcond}) with the estimate $1-x \leq \min\{1,-\log(x)\}$ for $x \in [0,1]$ and the substitution $t=1/u$ that
\begin{align}
\int_{\F} \int_{0}^{\infty}1-F(t)\,\mathrm{d}t \,\tilde{\mu}(\mathrm{d}F) &\leq \int_{\F} \int_{0}^{\infty}\min\{1,-\log(F(t))\}\,\mathrm{d}t \tilde{\mu}(\mathrm{d}F) \nonumber\\
& = \int_{\mathfrak{D}_+} \int_0^{\infty}\frac{\min\{1,f(u)\}}{u^2}\,\mathrm{d}u \,\gamma(\mathrm{d}f) <\infty,
\label{finite_mean}
\end{align}
i.e.\ $\tilde{\mu}$ is an element of $M_+^1({\F})$. For arbitrary $F \in {\F}$ we denote by $M_F:=\int_0^{\infty}1-F(t)\,\mathrm{d}t$ its mean. By (\ref{finite_mean}), the positive random variable $M_{\tilde{F}}$ has finite mean $c > 0$ (note that $M_{\tilde{F}}$ is positive almost surely and $c=0$ is ruled out since $f$ is non-vanishing, hence $\tilde{F}$ not almost surely identically equal to one), i.e.\ $\int_{{\F}} M_{F} \,\tilde{\mu}(\mathrm{d}F)=c $. Consequently, 
\begin{gather*}
\hat{\mu}(\mathrm{d}F) := \frac{M_F}{c}\,\tilde{\mu}(\mathrm{d}F)
\end{gather*}
defines an equivalent probability measure on $\F$. Finally, we denote by $\mu$ the probability measure that describes the law of the process $\{\tilde{F}_{M_{\tilde{F}}\,t}\}_{t \geq 0}$ under the measure $\hat{\mu}$, and observe that
\begin{gather*}
M_{\tilde{F}_{M_{\tilde{F}}\,.}} = \int_0^{\infty}1-\tilde{F}_{M_{\tilde{F}}\,s}\,\mathrm{d}s = \frac{1}{M_{\tilde{F}}}\,\int_0^{\infty}1-\tilde{F}_{s}\,\mathrm{d}s =1
\end{gather*}
almost surely. Consequently, $\mu \in M_+^1(\F_1)$. Putting together the pieces, we may re-write (\ref{molchanov_repr}) as
\begin{gather*} 
\{H_t\}_{t \geq 0} \stackrel{d}{=} \Big\{b\,t+\sum_{k \geq 1}-\log\Big[\tilde{F}^{(k)}_{\frac{\epsilon_1+\ldots+\epsilon_k}{t}-}\Big]\Big\}_{t \geq 0},\quad t \geq 0,
\end{gather*}
and we observe\footnote{Here, $\delta_{e}$ denotes the Dirac measure at a point $e$ in some Hausdorff space $\mathfrak{E}$.} that $\sum_{k \geq 1}\delta_{(\epsilon_1+\ldots+\epsilon_k,\tilde{F}^{(k)})}$ is a Poisson random measure on $[0,\infty) \times \F$ with mean measure $\mathrm{d}x \times \tilde{\mu}(\mathrm{d}F)$. Hence, the Laplace functional formula \cite[Proposition 3.6]{resnick87}, applied in the third equality below, gives
\begin{align*}
\IP(\vec{Y}>\vec{t}) & = \IE\Big[ e^{-\sum_{i \geq 1}H_{t_i}}\Big] = e^{-b\,\ell_{\Pi}(\vec{t})}\,\IE\Big[ \exp\Big\{{-\sum_{k \geq 1} -\log\big[\prod_{i \geq 1} \tilde{F}^{(k)}_{\frac{\epsilon_1+\ldots+\epsilon_k}{t_i}-}\big]}\Big\}\Big]\\
& = e^{-b\,\ell_{\Pi}(\vec{t})}\,\exp\Big\{ -\int_{{\F}}\int_0^{\infty}1-\prod_{i \geq 1}F\Big( \frac{x}{t_i}\Big)\,\mathrm{d}x\,\tilde{\mu}(\mathrm{d}F)\Big\}\\
& = e^{-b\,\ell_{\Pi}(\vec{t})}\,\exp\Big\{ -\int_{{\F}}\int_0^{\infty}1-\prod_{i \geq 1}F\Big( \frac{M_F\,x}{t_i}\Big)\,\mathrm{d}x\,M_F\,\tilde{\mu}(\mathrm{d}F)\Big\}\\
& = e^{-b\,\ell_{\Pi}(\vec{t})}\,\exp\Big\{ -c\,\int_{{\F}_1}\int_0^{\infty}1-\prod_{i \geq 1}F\Big( \frac{M_F\,x}{t_i}\Big)\,\mathrm{d}x\,{\hat{\mu}}(\mathrm{d}F)\Big\}\\
& = e^{-b\,\ell_{\Pi}(\vec{t})}\,\exp\Big\{ -c\,\int_{{\F}_1}\int_0^{\infty}1-\prod_{i \geq 1}F\Big( \frac{x}{t_i}\Big)\,\mathrm{d}x\,{{\mu}}(\mathrm{d}F)\Big\}\\
& = \exp\Big\{ -b\,\ell_{\Pi(\vec{t})}-c\,\int_{{\F}_1}\ell_F(\vec{t})\,\mu(\mathrm{d}F)\Big\}.
\end{align*}
The normalizing assumption $\IE[Y_1]=1$ in the definition of $\mathfrak{L}$ means that the exponential rate of the exponential random variable $Y_1$ equals one, which implies that $\ell(1,0,0,\ldots)=1$. We thus observe from the last equation that 
\begin{gather*}
1 = \ell(1,0,0,\ldots) = b+c.
\end{gather*} 
From this we conclude that $c=1-b$, hence
\begin{gather*}
\ell = b\,\ell_{\Pi}+(1-b)\,\int_{{\F}_1}\ell_{F}\,\mu(\mathrm{d}F),
\end{gather*}
as claimed. 
\end{proof}

The following corollary is of particular relevance when thinking about potential further research concerning the parameter estimation of non-decreasing strong IDT processes or exchangeable max-stable sequences, resp.\ extreme-value copulas.

\begin{corollary}[Uniqueness]\label{cor_uni}
The pair $(b,\mu)$ in Theorem \ref{thm_main} is unique.
\end{corollary}
\begin{proof}
It is convenient to study uniqueness in terms of the probability law of the uniquely associated strong IDT process $H$, determined by
\begin{gather*}
\IE\Big[ e^{-\sum_{k \geq 1}H_{t_k}}\Big] = e^{-b\,\ell_{\Pi}(\vec{t})-(1-b)\,\int_{\F_1}\ell_F(\vec{t})\,\mu(\mathrm{d}F)},\quad \vec{t} \in [0,\infty)^{\IN}_{00}.
\end{gather*}
The constant $b$ is unique, because it is the unique drift constant in the L\'evy-Khinchin representation of the infinitely divisible random variable $H_1$. To explain this, denote by $\bm{1}_n=(1,\ldots,1)$ an $n$-dimensional row vector with all entries equal to one. For each fixed $F \in \F_1$ the stable tail dependence function $\ell_F$ satisfies $\ell_F(\bm{1}_n,0,0,\ldots)=\Psi_F(n)$ for a Bernstein function $\Psi_F$ without drift, see \cite[Lemma 3]{mai17}. This property carries over to probability mixtures, so that $\IE[\exp(-n\,H_1)]=\exp(-b\,n-(1-b)\,\Psi(n))$ for some Bernstein function $\Psi$ without drift. Hence, the Bernstein function associated with $H_1$ has (unique) drift $b$ and its L\'evy measure equals $1-b$ times the L\'evy measure of $\Psi$.
\par
Regarding $\mu$, it follows from Lemma \ref{lemma_molchanov} and the proof of Theorem \ref{thm_main} that the probability law of $f$ in any Le Page series representation for $H$ is necessarily supported by the set 
\begin{align*}
\mathfrak{G}&:=\Big\{ g:[0,\infty) \rightarrow [0,\infty]\,:\,g(0)=0,\,g\mbox{ right-continuous, non-decreasing,}\\
& \qquad \qquad \lim_{t \rightarrow \infty}g(t)=\infty,\,\int_0^{\infty}1-e^{-g_{\frac{1}{s}}}\,\mathrm{d}s<\infty\Big\}.
\end{align*}
For each $g \in \mathfrak{G}$ there is a unique $c>0$, namely
\begin{gather*}
c:=\int_0^{\infty}1-e^{-g_{\frac{1}{s}}}\,\mathrm{d}s,
\end{gather*}
such that\footnote{We use the notation of \cite{molchanov18}, denoting $(c \circ g^{(1)})(t):=g^{(1)}(c\,t)$, for $t \geq 0$ and $c>0$.} $g = c \circ g^{(1)}$, where $g^{(1)}$ lies in the smaller set 
\begin{gather*}
\mathfrak{G}_1 := \Big\{g \in \mathfrak{G}\,:\,\int_0^{\infty}1-e^{-g_{\frac{1}{s}}}\,\mathrm{d}s=1\Big\}.
\end{gather*}
By \cite[Remark 4.1]{molchanov18} the law of $f^{(1)}$ is unique. This implies that the measure $\mu$ is unique, since by the proof of Theorem \ref{thm_main} it equals the law of 
\begin{gather*}
F_t:=e^{-\lim_{x \downarrow t} f^{(1)}_{1/x}},\quad t \geq 0,
\end{gather*}
and this transformation maps $\mathfrak{G}_1$ to $\F_1$ in a bijective manner. 
\end{proof}

\begin{remark}[Consequences and explanations of the results]\label{rmk}
We collect a few explanatory remarks ((b),(c),(f)) and immediate consequences ((a),(d),(e)) of the main results:
\begin{itemize}
\item[(a)] \textbf{$\mathfrak{L}$ is a simplex:}\\
We first provide a proof that $\mathfrak{L}$ is compact (in the topology of pointwise convergence).
\begin{proof}
Let $\{\ell_n\}_{n \in \IN} \subset \mathfrak{L}$. Then we find strong IDT processes $\{H^{(n)}\}_{n \in \IN}$ associated with these $\ell_n$. The processes $F^{(n)}:=1-\exp(-H^{(n)})$ define random variables on the space of distribution functions of non-negative random variables. The set of distribution functions of random variables taking values in $[0,\infty]$ (equipped with the topology of pointwise convergence at all continuity points of the limit) is compact by Helly's Selection Theorem and Hausdorff (since it is metrizable by the L\'evy metric, see \cite{sibley71}). Thus, the Radon probability measures on this set (equipped with the weak topology) form a Bauer simplex by \cite[Corollary II.4.2, p.\ 104]{alfsen71}, in particular form a compact set. Since the probability measures of the given sequence $\{F^{(n)}\}$ lie in this set, we find a convergent subsequence $\{F^{(n_i)}\}$ and a limiting law, hence a limiting stochastic process $F$, and we define $H:=-\log(1-F)$. Then $F \in M_+^{1}(\mathfrak{F})$ almost surely, which follows from 
\begin{align*}
\IE\Big[ \int_0^{\infty}1-F_s\,\mathrm{d}s\Big] &= \IE\Big[ \int_0^{\infty}e^{-H_s}\,\mathrm{d}s\Big]= \int_0^{\infty}\IE\Big[e^{-H_s}\Big]\,\mathrm{d}s\\
& = \int_0^{\infty}\lim_{i \rightarrow \infty}\IE\Big[e^{-H^{(n_i)}_s}\Big]\,\mathrm{d}s = \int_0^{\infty}e^{-s}\,\mathrm{d}s = 1<\infty,
\end{align*}
where the third equality follows from the bounded convergence theorem and the fourth from the fact that $\IE[\exp(-H^{(n)}_s)]=\exp(-s\,\ell_n(1,0,0,\ldots))=\exp(-s)$ for each $n \in \IN$. We define
\begin{gather*}
\ell(\vec{t}):=-\log\Big\{ \IE\Big[ e^{-\sum_{k \geq 1}H_{t_k}}\Big]\Big\},\quad \vec{t} \in [0,\infty)^{\IN}_{00},
\end{gather*}
and claim that $\ell$ equals the limit of $\ell_{n_i}$. To see this, for fixed $\vec{t}$ we compute
\begin{align*}
\ell(\vec{t}) & = -\log\Big\{ \IE\Big[ e^{-\sum_{k \geq 1}H_{t_k}}\Big]\Big\} = -\log\Big\{ \IE\Big[ e^{-\sum_{k \geq 1}\lim_{i \rightarrow \infty}H^{(n_i)}_{t_k}}\Big]\Big\}\\
& \stackrel{(\ast)}{=} -\log\Big\{ \IE\Big[ e^{-\lim_{i \rightarrow \infty}\sum_{k \geq 1}H^{(n_i)}_{t_k}}\Big]\Big\} \stackrel{(\ast\ast)}{=}  -\log\Big\{\lim_{i \rightarrow \infty} \IE\Big[ e^{-\sum_{k \geq 1}H^{(n_i)}_{t_k}}\Big]\Big\} \\
& =  -\log\Big\{\lim_{i \rightarrow \infty}e^{-\ell_{n_i}(\vec{t})}\Big\} = \lim_{i \rightarrow \infty}\ell_{n_i}(\vec{t}).
\end{align*}
We have used the fact that almost all entries of $\vec{t}$ are zero in $(\ast)$ and bounded convergence in $(\ast\ast)$. Finally, to see that $H$ is strong IDT, it suffices to verify that the homogeneity of order one of the $\ell_{n_i}$ carries over to the limit $\ell$ (obviously), hence $\ell \in \mathfrak{L}$.
\end{proof}
$\mathfrak{L}$ is obviously convex and by Theorem \ref{thm_main} the extremal boundary of $\mathfrak{L}$ is $\partial_e \mathfrak{L}=\{\ell_{\Pi}\}\cup \{\ell_F\,:\,F \in \F_1\}$. Thus, $\mathfrak{L}$ is a simplex, since the boundary integral representation is unique by Corollary \ref{cor_uni}. Whether $\mathfrak{L}$ is a Bauer simplex, i.e.\ whether $\partial_e \mathfrak{L}$ is closed, is not obvious, since $\F_1$ is not compact. It is left as an open question at this point. 
\item[(b)] \textbf{The isolated nature of $\ell_{\Pi}$:}\\
One noticeable aspect about the topology on $\mathfrak{L}$ is that the seemingly isolated point $\ell_{\Pi}$ is in fact not isolated. A sequence $\{\ell_{F_n}\}_{n \in \IN} \subset \partial_e \mathfrak{L}$ converges (pointwise) to $\ell_{\Pi}$ if and only if $\{F_n\}_{n \in \IN} \subset \F_1$ converges to $F_0$ at all continuity points of $F_0$ (which means at all $x>0$, but not necessarily at $x=0$). To see this, we point out for $F \in \F_1$ that
\begin{gather*}
\ell^{(2)}_{F_n}(1,1) = 2-\int_0^{\infty}\big(1-F_n(s)\big)^2\,\mathrm{d}s=2-||F_0-F_n||_{L^2}^{2}.
\end{gather*}
Thus, $\ell^{(2)}_{F_n}(1,1)$, which is always $\leq 2$, converges to $2$ as $n \rightarrow \infty$ if and only if $F_n(x)$ converges to $F_0(x)=1$ for all $x>0$. But the only element $\ell \in \mathfrak{L}$ satisfying $\ell^{(2)}(1,1) = 2$ is $\ell_{\Pi}$, since $\ell^{(2)}_F(1,1)<2$ for each $F \in \F_1$.
\item[(c)] \textbf{Probabilistic interpretation of the results:}\\ 
As already remarked in the introduction, rewriting Theorem \ref{thm_main} and Corollary \ref{cor_uni} in probabilistic terms, instead of in the analytical terms of $\mathfrak{L}$, means that $\ell \in \mathfrak{L} \setminus\{\ell_{\Pi}\}$ has the spectral representation
\begin{gather}
\ell(\vec{t}) = b\,\ell_{\Pi}(\vec{t})+(1-b)\,\IE\big[ \max_{k \geq 1}\{t_k\,X_k\}\big],\quad \vec{t} \in [0,\infty)^{\IN}_{00},
\label{specshort}
\end{gather}
where $b \in [0,1)$ and $\vec{X}$ is an exchangeable sequence of non-negative random variables with $\IE[X_1]=1$ and the property that $\lim_{n \rightarrow \infty} \frac{1}{n}\,\sum_{k=1}^{n}X_k$ is almost surely identically equal to one. This representation is canonical in the sense that the constant $b$ as well as the probability law of $\vec{X}$ are unique.
\par
Recalling the classical extreme-value theory based on Poisson random measure, see \cite{resnick87} for a textbook treatment, a stochastic representation for $\vec{Y}$ based on the spectral representation (\ref{specshort}) is given by
\begin{gather*}
\vec{Y} \stackrel{d}{=} \Big( \min\Big\{\frac{ Y_1^{(1)}}{b},\frac{ Y_1^{(0)}}{1-b}\Big\},\,\min\Big\{\frac{ Y_2^{(1)}}{b},\frac{ Y_2^{(0)}}{1-b}\Big\},\ldots\Big),
\end{gather*}
where $\vec{Y}^{(1)}=(Y^{(1)}_1,Y^{(1)}_2,\ldots)$ is a sequence of iid unit exponentials (corresponding to the case $b=1$), and, independently, $\vec{Y}^{(0)}=(Y^{(0)}_1,Y^{(0)}_2,\ldots)$ corresponds to the case $b=0$ and satisfies 
\begin{gather*}
\vec{Y}^{(0)} \stackrel{d}{=} \Big(\min_{n \geq 1}\Big\{ \frac{\epsilon_1+\ldots+\epsilon_n}{X_1^{(n)}}\Big\},\,\min_{n \geq 1}\Big\{ \frac{\epsilon_1+\ldots+\epsilon_n}{X_2^{(n)}}\Big\},\ldots\Big),
\end{gather*}
where $\vec{X}^{(n)}$ are independent copies of $\vec{X}$ and, independently, $\epsilon_1,\epsilon_2,\ldots$ is an iid sequence of unit exponentials. This classical representation does not make explicit use of exchangeability, but is only an instance of the general (non-exchangeable) theory. An alternative stochastic representation for $\vec{Y}$, due to \cite[Theorem 5.3]{maischerer13} and making use of exchangeability, is given by (\ref{ciiddef}) with the stochastic process $H$ defined via its Le Page representation
\begin{gather*}
H_t = b\,t+(1-b)\,\sum_{k \geq 1}-\log\Big\{ F^{(k)}_{\frac{\epsilon_1+\ldots+\epsilon_k}{t}-} \Big\},
\end{gather*}
where $F^{(k)}$ are independent copies of the random distribution function $F$ from which $\vec{X}$ is drawn. By Glivenko-Cantelli, $F$ may be written in terms of $\vec{X}$ as $F_t = \lim_{n \rightarrow \infty}\frac{1}{n}\,\sum_{i=1}^{n}1_{\{X_i \leq t\}}$. In other words, the sequence $\vec{Y}$ is an iid sequence drawn from the random distribution function $t \mapsto 1-\exp(-H_t)$. Notice in particular that $M_F=\int_0^{\infty}1-F_t\,\mathrm{d}t=\lim_{n \rightarrow \infty}\frac{1}{n}\,\sum_{i=1}^{n}X_i$ is identically one by the normalization to $\F_1$ (instead of $\F$). Depending on the law of $\vec{X}$ (and thus the properties of $H$), this representation is particularly convenient to simulate the random vector $(Y_1,\ldots,Y_d)$ for arbitrarily large $d$. An alternative strategy to accomplish this simulation, whose idea rather makes use of the general (non-exchangeable) theory, is presented in the following bullet point (d).
\item[(d)] \textbf{Pickands representation of finite-dimensional margins:}\\
The Pickands representation of $\ell^{(d)}_F$ has been derived in \cite[Lemma 4]{mai17}. Furthermore, the Pickands representation of $\ell^{(d)}_{\Pi}$ is well known to correspond to a uniform distribution on $\{1,\ldots,d\}$. Combining these facts with Theorem \ref{thm_main} immediately implies for the random vector $\vec{X}^{(d)}$ in (\ref{Pickandsrepr}) associated with $\ell^{(d)}$ for $\ell \in \mathfrak{L}$, represented by $(b,\mu)$, that
\begin{gather}
\vec{X}^{(d)} \stackrel{d}{=} \Big(\frac{W^{(d)}_1}{\sum_{i=1}^{d}W^{(d)}_i},\ldots, \frac{W^{(d)}_d}{\sum_{i=1}^{d}W^{(d)}_i} \Big),
\label{repr_Pick_ell}
\end{gather}
where the random vector $\vec{W}^{(d)}$ can be simulated as follows:
\begin{itemize}
\item Draw a random variable $D$ which is uniformly distributed on $\{1,\ldots,d\}$.
\item Draw a Bernoulli random variable with success probability $b$. If success, define $W^{(d)}_k:=1_{\{k=D\}}$ for $k=1,\ldots,d$ and return. Otherwise, proceed with the following steps.
\item Simulate the random distribution function $F=\{F_t\}_{t \geq 0}$ from the probability measure $\mu \in M_+^1(\F_1)$ and draw a random variable $Z$ with distribution function $t \mapsto \int_0^{t}s\,\mathrm{d}F_s$, $t \geq 0$.
\item Draw iid random variables $Z_1,\ldots,Z_d$ with distribution function $F$.
\item Define $W^{(d)}_k:=1_{\{k=D\}}\,Z+1_{\{k \neq D\}}\,Z_k$ for $k=1,\ldots,d$ and return.
\end{itemize}
This algorithm to simulate the random vector $\vec{X}^{(d)}$ can be used to derive an exact simulation algorithm for the random vector $(Y_1,\ldots,Y_d)$, see \cite[Algorithm 1]{dombry16}.
\par 
We emphasize again at this point that the probability law of $\vec{X}^{(d)}$ is uniquely determined by the function $\ell^{(d)}$. However, there exist exchangeable random vectors $\vec{X}^{(d)}$ taking values in $S_d$, and thus via (\ref{Pickandsrepr}) lead to symmetric $d$-dimensional stable tail dependence functions, but which cannot be represented as in (\ref{repr_Pick_ell}). These correspond to $d$-dimensional exchangeable max-stable random vectors that do not satisfy the stronger notion of ``infinite (De Finetti) exchangeability'' (or ``conditionally iid''), and are thus outside the realm of the present article since they do not arise as $d$-margins of some $\ell \in \mathfrak{L}$. An example for $d=3$ can be retrieved from \cite[Example 2.4]{maibrazil13}. We recall from this example that the $3$-variate stable tail dependence function
\begin{gather*}
\ell^{(3)}(t_1,t_2,t_3) = \frac{\lambda_1}{\lambda_1+2\,\lambda_2+\lambda_3}\,t_{[1]}+\frac{\lambda_1+\lambda_2}{\lambda_1+2\,\lambda_2+\lambda_3}\,t_{[2]}+t_{[3]},
\end{gather*}
with $t_{[1]} \leq t_{[2]} \leq t_{[3]}$ the order list of $t_1,t_2,t_3$ and with positive parameters $\lambda_1,\lambda_2,\lambda_3>0$, arises as $3$-margin of some $\ell \in \mathfrak{L}$ if and only if $\lambda_2^2 \leq \lambda_1\,\lambda_3$. If this condition is violated, $\ell^{(3)}$ is still a symmetric stable tail dependence function, but the associated random vector $(Y_1,Y_2,Y_3)$ cannot have a stochastic representation that is ``conditionally iid''.
\item[(e)] \textbf{Arbitrary, non-decreasing strong IDT processes:}\\
The probability law of a non-decreasing, right-continuous, strong IDT process $H=\{H_t\}_{t \geq 0}$, which is not deterministic (i.e.\ not identically $H_t=b\,t$ for some $b \geq 0$), is uniquely described by a triplet $(b,c,\mu) \in [0,\infty)\times (0,\infty) \times M_+^{1}(\F_1)$, and has the LePage series representation
\begin{gather*}
H_t=b\,t+c\,\sum_{k \geq 1}-\log\Big\{ F^{(k)}_{\frac{\epsilon_1+\ldots+\epsilon_k}{t}-}\Big\},
\end{gather*}
where $\sum_{k \geq 1}\delta_{(\epsilon_1+\ldots+\epsilon_k,F^{(k)})}$ is a Poisson random measure on $[0,\infty) \times \F_1$ with mean measure $\mathrm{d}t \times \mu$. Lemma \ref{lemma_molchanov} and the proof of Theorem \ref{thm_main} show that other LePage series representations of the form (\ref{molchanov_repr}) can only differ from this canonical one by changing from the unique measure $\mu \in M_+^{1}(\F_1)$ to some $\tilde{\mu} \in M_+^{1}(\F)$ (and potentially adjusting the constant $c$ accordingly). The unboundedness of $b$ as well as the additional constant $c>0$ in the triplet $(b,c,\mu)$, when compared to $(b,\mu)$ in Theorem \ref{thm_main}, is due to the fact that a stable tail dependence function $\ell$ is normalized to satisfy $\ell(1,0,0,\ldots)=b+c=1$ (corresponding to $\IE[Y_1]=1$), while the triplet $(b,c,\mu)$ describes the law of an arbitrary non-decreasing, right-continuous strong IDT process $H$ via the relation
\begin{gather*}
\IE\Big[ e^{-\sum_{k \geq 1}H_{t_k}}\Big] = e^{-b\,\ell_{\Pi}(\vec{t})-c\,\int_{\F_1}\ell_F(\vec{t})\,\mu(\mathrm{d}F)}.
\end{gather*}
\item[(f)] \textbf{The measure change in Theorem \ref{thm_main}:}\\
If $\tilde{F}$ is a random variable on $(\Omega,\G,\IP)$ taking values in $\F$ with the additional property that $\IE[M_{\tilde{F}}]=1$, with $M_{\tilde{F}}=\int_0^{\infty}1-\tilde{F}_s\,\mathrm{d}s$, then 
\begin{gather*}
\ell(\vec{t}) := \IE\Big[ \int_0^{\infty}1-\prod_{k \geq 1}\tilde{F}_{\frac{s}{t_k}}\,\mathrm{d}s\Big]
\end{gather*}
defines an element in $\mathfrak{L}$. What is its canonical boundary integral representation? To this end, we define $F_t:=\tilde{F}_{M_{\tilde{F}}\,t}$, $t \geq 0$, and observe that $F$ takes values in $\F_1$. We further define an equivalent probability measure $\IQ$ via the measure change $\mathrm{d}\IQ = M_{\tilde{F}}\,\mathrm{d}\IP$. Denoting expectation with respect to $\IQ$ by $\IE^{\IQ}$, we observe that
\begin{align*}
\ell(\vec{t}) &= \IE\Big[ \int_0^{\infty}1-\prod_{k \geq 1}\tilde{F}_{\frac{s}{t_k}}\,\mathrm{d}s\Big] = \IE\Big[M_{\tilde{F}}\, \int_0^{\infty}1-\prod_{k \geq 1}\tilde{F}_{M_{\tilde{F}}\,\frac{s}{t_k}}\,\mathrm{d}s\Big]\\
&=\IE\big[M_{\tilde{F}}\,\ell_F(\vec{t})\big]=\IE^{\IQ}\big[ \ell_{F}(\vec{t})\big].
\end{align*}
Example \ref{ex_logistic} below provides an example for such $\ell$.
\end{itemize}
\end{remark}

We end this section with a few examples.

\begin{example}[The L\'evy subordinator case]\label{exlevy}
Recall that a \emph{L\'evy subordinator} $L=\{L_t\}_{t \geq 0}$, see \cite{bertoin99} for a textbook treatment, is a non-decreasing, right-continuous strong IDT process with independent and stationary increments, which implies that its law is fully determined by the law of $L_1$. By the well known L\'evy-Khinchin formula for infinitely divisible distributions, the probability law of $L_1$ is canonically described in terms of a pair $(b_L,\nu_L)$ of a drift constant $b_L \geq 0$ and a Radon measure $\nu_L$ on $(0,\infty]$ subject to the condition $\int_0^{1}x\,\nu_L(\mathrm{d}x)<\infty$, the so-called L\'evy measure. The (non-deterministic) L\'evy subordinator $L$ associated with the pair $(b_L,\nu_L)$ is obtained when specifying $b:=b_L$, $c:=\int_{(0,\infty]}1-\exp(-x)\,\nu_L(\mathrm{d}x)$, and $\mu$ as the law of the random distribution function
\begin{gather*}
F_t:=e^{-\Theta}+\big(1-e^{-\Theta} \big)\,1_{\{1-e^{-\Theta} \geq 1/t\}},\quad t \geq 0, \quad \Theta \sim \big(1-e^{-x}\big)\,\nu_L(\mathrm{d}x)/c.
\end{gather*}
Conditioned on the randomized parameter $\Theta$, this $F$ corresponds to a random variable taking the value $1/(1-\exp(-\Theta))$ with probability $1-\exp({-\Theta})$, and the value zero with complementary probability $\exp(-\Theta)$. The random parameter $\Theta$ itself is drawn from the probability measure $\big(1-e^{-x}\big)\,\nu_L(\mathrm{d}x)/c$. Notice that every probability measure on $(0,\infty]$ is possible for $\Theta$, but the law of $\Theta$ is invariant with respect to changes of $c$, that is when changing from $\nu_L$ to $\beta\,\nu_L$ for some $\beta>0$. 
\par
If $L$ is normalized to satisfy $b+c=1$, this example corresponds to a stable tail dependence function $\ell \in \mathfrak{L}$, given by 
\begin{align*}
\ell(\vec{t}) &= \sum_{k=1}^{d(\vec{t})}t_{[k]}\,\big(\Psi(d(\vec{t})-k+1)-\Psi(d(\vec{t})-k)\big),\\
\Psi(x)&=b\,x+\int_{(0,\infty]}1-e^{-x\,t}\,\nu_L(\mathrm{d}t),
\end{align*}
where $d(\vec{t}):=\max\{n \in \IN\,:\,t_n>0\}$ and $t_{[1]} \leq t_{[2]} \leq \ldots \leq t_{[d(\vec{t})]}$ denotes an ordered list of $t_1,\ldots,t_{d(\vec{t})}$. The associated extreme-value copulas $C_{\ell}$ form precisely the ``conditionally iid'' subfamily of the survival copulas of the Marshall-Olkin exponential distribution, see \cite[Chapter 3.3]{maischerer17} for a textbook treatment of this connection. Furthermore, $\ell \in \partial_e \mathfrak{L}$ if and only if either $L_t=t$ (corresponding to $b=1$ and $\ell=\ell_{\Pi}$) or if $L$ is a compound Poisson subordinator with constant jump sizes. In the latter case, $b=0$ and $\ell=\ell_F$ with $F$ as described above, but $\Theta$ a non-random, positive constant. 
\end{example}

\begin{example}[(Generalized) logistic model]\label{ex_logistic}
Let $\ell \in \mathfrak{L}$ arbitrary and $\alpha \in (0,1)$. Furthermore, we denote by $H=\{H_t\}_{t \geq 0}$ a non-decreasing strong IDT process associated with $\ell$, and by $H^{(k)}$ independent copies thereof, $k \in \IN$. Let $M$ be a positive random variable with Laplace transform $\IE[\exp(-x\,M)]=\exp(-x^{\alpha})$, i.e.\ an $\alpha$-stable random variable, independent of $H$. It is not difficult to see that $\{H_{M\,t^{1/\alpha}}\}_{t \geq 0}$ is another non-decreasing strong IDT process, and its associated stable tail dependence function is given by
\begin{gather*}
\ell_{\alpha}(\vec{t})=\ell\big(t_1^{1/\alpha},\,t_2^{1/\alpha},\ldots\big)^{\alpha},\quad \vec{t}=(t_1,\,t_2,\ldots) \in [0,\infty)^{\IN}_{00},
\end{gather*} 
satisfying $\ell_{\alpha} \in \mathfrak{L}$, since $\ell_{\alpha}(1,0,0,\ldots)=1$. With the constant $c_{\alpha}:=\Gamma(1-\alpha)^{-1/{\alpha}}$, a Le Page series representation for this stochastic process is given by
\begin{gather*}
\{H_{M\,t^{1/\alpha}}\}_{t \geq 0} \stackrel{d}{=}\Big\{\sum_{k \geq 1}-\log\big( \tilde{F}^{(k)}_{\frac{\epsilon_1+\ldots+\epsilon_k}{t}-}\big)\Big\}_{t \geq 0},
\end{gather*}
where $\tilde{F}^{(k)}$ are independent copies of $\tilde{F}_t:= \exp(-H_{c_{\alpha}\,t^{-1/\alpha}})$. We notice that each realization of $\tilde{F}$ is an element of $\F$ (though not necessarily of $\F_1$), since
\begin{gather*}
\IE\Big[ \int_0^{\infty}1-\tilde{F}_t\,\mathrm{d}t\Big] = \int_0^{\infty}1-e^{-c_{\alpha}\,t^{-1/\alpha}}\,\mathrm{d}t=1<\infty.
\end{gather*}
The canonical Le Page series representation is obtained when changing from $\tilde{F}$ to $F$, where $F_t := \tilde{F}_{M_{\tilde{F}}\,t}$, and additionally changing measure like in Remark \ref{rmk}(f). As a final remark, if $\ell=\ell_{\Pi}$, meaning $H_t=t$, then $\ell_{\alpha}=\ell_{F_{\alpha}} \in \partial_e \mathfrak{L}$ corresponds precisely to the well known logistic model based on the Fr\'echet distribution function $F_{\alpha}(x)=\exp(-c_{\alpha}\,x^{-1/\alpha}) \in \F_1$.
\end{example}

\begin{example}[A convenient umbrella for many well known models]
As already highlighted in \cite{mai17}, some well known models correspond to extremal points $\ell_F \in \partial_e \mathfrak{L}$, for instance the logistic model (see previous example) and the negative logistic model. However, Example \ref{exlevy} shows that the popular Marshall-Olkin model, resp.\ the infinite exchangeable subfamily thereof, is not an extremal element in general. With the motivation to establish an analytically tractable umbrella for many well known parametric models, in particular including both $\partial_e \mathfrak{L}$ and Example \ref{exlevy}, a rich semi-parametric specification for $\ell \in \mathfrak{L}$ can be constructed as follows. For $F \in \F_1$ the function $\Psi_F(z):=\int_0^{\infty}1-F^z(t)\,\mathrm{d}t$ defines a Bernstein function with $\Psi_F(1)=1$, see \cite[Lemma 3]{mai17}. This implies that for arbitrary $z \in (0,\infty)$ the function $F_z(x):= F(x\,\Psi_F(z))^z$ lies again in $\F_1$. With a probability law $\rho$ on $(0,\infty)$ we see $\ell_{\rho,F} := \int_{(0,\infty)}\ell_{F_z}\,\rho(\mathrm{d}z) \in \mathfrak{L}$. Many parametric models from the literature are comprised by this construction. In particular, the elements $\ell_F \in \partial_e \mathfrak{L}$ correspond to $\rho = \delta_{1}$ by construction, and Example \ref{exlevy} corresponds to the special case when $F(x)=\exp(-1)+(1-\exp(-1))\,1_{\{x \geq 1/(1-\exp(-1))\}}$ is held fix, but $\rho$ is varied, corresponding to the law of $\Theta$.
\end{example}

\begin{example}[Constructing $\ell$ via inclusion-exclusion]
Let $\ell_X \in \mathfrak{L}$ arbitrary and assume that $\vec{X}$ is min-stable multivariate exponential with stable tail dependence function $\ell_X$. Now consider $\vec{Y}$ with a spectral representation given in terms of $\vec{X}$, i.e.\ $-\log \{\IP(\vec{Y}>\vec{t})\}=\IE[\max_{k \geq 1}\{t_k\,X_k\}]$. Using the principle of inclusion and exclusion it is not difficult to compute the stable tail dependence function $\ell_Y$ of $\vec{Y}$, to wit
\begin{gather*}
\ell_Y(\vec{t}) = \sum_{k=1}^{d(\vec{t})}(-1)^{k+1}\,\sum_{1 \leq i_1<\ldots<i_k \leq d(\vec{t})}\ell_X\Big(\frac{1}{t_{i_1}},\ldots,\frac{1}{t_{i_k}},0,0,\ldots\Big)^{-1},
\end{gather*}
with $d(\vec{t}):=\max\{n \in \IN\,:\,t_n>0\}$ as in Example \ref{exlevy}. In particular, if $\ell_X=\ell_{F_{\alpha}}$ with $F_{\alpha}$ from Example \ref{ex_logistic} (logistic model), then $\ell_Y$ corresponds to a negative logistic model. The mapping $\ell_X \mapsto \ell_Y$ on $\mathfrak{L}$ seems to be ``association-increasing''. For instance, we observe that $\ell_Y^{(2)}(1,1)=2-\ell_X^{(2)}(1,1)^{-1} \leq \ell_X^{(2)}(1,1)$. Furthermore, maximal dependence $\ell_X(\vec{t})=\ell_Y(\vec{t})=\max_{k \geq 1}\{t_k\}$ is a fixpoint. It might potentially be interesting to study the relationship between the spectral representations of $\ell_X$ and $\ell_Y$.
\end{example}

\section{Conclusion}\label{sec_concl}
It has been shown that the set $\mathfrak{L}$ of infinite-dimensional, symmetric stable tail dependence functions is a simplex. The boundary of the simplex and a respective boundary integral representation for $\ell \in \mathfrak{L}$ has been derived in terms of a pair $(b,\mu)$ of a constant $b\in [0,1]$ and a probability measure $\mu$ on the set of distribution functions of non-negative random variables with unit mean. Equivalently, the pair $(b,\mu)$ was shown to conveniently describe the probability law of a non-decreasing, right-continuous stochastic process which is strongly infinitely divisible with respect to time, subject to a normalizing condition.

\section*{Acknowledgments}
Inspiring discussions with Paul Ressel and his helpful comments on earlier versions of this manuscript are gratefully acknowledged. His remarks in particular made me aware of the somewhat special role of the point $\ell_{\Pi}$. Helpful comments by the anonymous referees and the handling editor are also gratefully acknowledged.


\begin{thebibliography}{}
\bibitem{aldous85} D.J.\ Aldous, Exchangeability and related topis, {Springer, \'Ecole d'\'Et\'e de Probabilit\'es de Saint-Flour XIII-1983. Lecture Notes in Mathematics} {1117} pp.\ 1--198 (1985).
\bibitem{alfsen71} E.M.\ Alfsen, Compact convex sets and boundary integrals, Springer (1971).
\bibitem{bertoin99} J.\ Bertoin, Subordinators: Examples and Applications, {in \'Ecole d'\'Et\'e de Probabilit\'es de Saint-Flour XXVII-1997, Lecture Notes in Mathematics, Springer} {1717} pp.\ 1--91 (1999).
\bibitem{definetti31} B.\ De Finetti, Funzione caratteristica di un fenomeno allatorio, {Atti della R.\ Accademia Nazionale dei Lincii Ser.\ 6, Memorie, Classe di Scienze, Fisiche, Matematiche e Naturali} {4} pp.\ 251--299 (1931).
\bibitem{definetti37} B.\ De Finetti, La pr\'evision: ses lois logiques, ses sources subjectives, {Annales de l'Institut Henri Poincar\'e} {7} pp.\ 1--68 (1937).
\bibitem{dehaan77} L.\ De Haan, S.I.\ Resnick, Limit theory for multivariate sample extremes,   Zeitschrift f\"ur Wahrscheinlichkeitstheorie und Verwandte Gebiete 40 pp.\ 317--337  (1977).
\bibitem{dehaan84} L.\ De Haan, A spectral representation for max-stable processes,  Annals of Probability 12  pp.\ 1194--1204 (1984).
\bibitem{dombry16} C.\ Dombry, S.\ Engelke, M.\ Oesting, Exact simulation of max-stable processes, {Biometrika} {103:2} pp.\ 303--317 (2016).
\bibitem{gudendorf09} G.\ Gudendorf, J.\ Segers, Extreme-value copulas, In: Copula Theory and its Applications, (P.\ Jaworski, F. Durante, W.K. H\"ardle, T. Rychlik, T. Eds.), Springer, New York 129--145  (2009).
\bibitem{joe97} H.\ Joe, Multivariate models and dependence concepts, {Chapman $\&$ Hall, Boca Raton} (1997).
\bibitem{mai17} J.-F.\ Mai, Extreme-value copulas associated with expected scaled maxima of independent random variables, {Journal of Multivariate Analysis} 166 pp.\ 50--61 (2018).
\bibitem{maischerer13} J.-F.\ Mai, M.\ Scherer, Characterization of extendible distributions with exponential minima via processes that are infinitely divisible with respect to time, {Extremes} {17} pp.\ 77--95 (2014).
\bibitem{maibrazil13} J.-F.\ Mai, M.\ Scherer, Extendibility of Marshall-Olkin distributions and inverse Pascal triangles, {Brazilian Journal of Probability and Statistics} {27} pp.\ 310--321 (2013).
\bibitem{maischerer17} J.-F.\ Mai, M.\ Scherer, Simulating Copulas, 2nd edition, {World Scientific, Singapore} (2017).
\bibitem{molchanov18} C.\ Kopp, I.\ Molchanov, Series representations of time-stable stochastic processes, Probab. Math. Statist. (2018, in press).
\bibitem{pickands81} J.\ Pickands, Multivariate extreme value distributions, Proceedings of the 43rd Session ISI, Buenos Aires pp.\ 859--878 (1981).
\bibitem{resnick87} S.\ I.\ Resnick, Extreme values, regular variation and point processes, {Springer-Verlag} (1987).
\bibitem{ressel13} P.\ Ressel, Homogeneous distributions and a spectral representation of classical mean values and stable tail dependence functions, Journal of Multivariate Analysis 117 pp.\ 246--256 (2013).
\bibitem{schilling10} R.\ Schilling, R.\ Song, Z.\ Vondracek, Bernstein Functions, De Gruyter (2010).
\bibitem{segers12} J.\ Segers, Max-stable models for multivariate extremes, {Statistical Journal} {10}, pp.\ 61--82  (2012).
\bibitem{sibley71} D.A.\ Sibley, A metric for weak convergence of distribution functions, {Rocky Mountain Journal of Mathematics} {1:3}, pp.\ 427--430 (1971).
\end{thebibliography}
\end{document}